\documentclass[runningheads]{llncs}
\usepackage{graphicx}

\usepackage{amssymb}
\usepackage{ifthen}
\usepackage{color,soul}
\usepackage[linesnumbered,ruled]{algorithm2e}
\usepackage{mathtools}
\usepackage{url}
\usepackage{tikz}
\usetikzlibrary{arrows,backgrounds,decorations,decorations.pathmorphing,
	positioning,fit,automata,shapes,patterns,plotmarks,calc,trees,
	datavisualization,arrows.meta}
\usepackage{xspace}
\usepackage{alltt}
\usepackage{wrapfig}
\usepackage{subcaption}
\usepackage{colortbl}
\usepackage{enumitem}
\usepackage{hyperref}
\hypersetup{%
  colorlinks=true,           
  allcolors=blue!70!black,   
  pdfstartview=Fit,          
  breaklinks=true,
  pdfauthor={MBAB},
  pdftitle={Distributed Online Detection of Continuous-Time Predicates}}
\usepackage{listings}
\usepackage{stmaryrd}
\usepackage{multirow}

\usepackage{pgfplots}
\usepackage{pgfplotstable}
\usepackage{graphicx}
\usepackage{todonotes}

\pgfplotsset{compat=1.18}

\usepackage{macros}

\newcommand{\ourslicer}{\mathcal{S}}
\newcommand{\ourabstractor}{\mathcal{A}}

\begin{document}

\title{Decentralized Predicate Detection over Partially Synchronous Continuous-Time Signals\thanks{Supported by NSF awards 2118179 and 2118356.}}
\titlerunning{Decentralized Continuous-Time Predicate Detection}

\def\orcidID#1{\smash{\href{http://orcid.org/#1}{\protect\raisebox{-1.25pt}{\protect\includegraphics{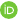}}}}}

\author{Charles Koll\inst{1}\orcidID{0000-0001-5941-250X} \and
Anik Momtaz\inst{2}\orcidID{0000-0002-4739-1032} \and
Borzoo Bonakdarpour\inst{2}\orcidID{0000-0003-1800-5419} \and
Houssam Abbas\inst{1,3}\orcidID{0000-0002-8096-2618}}
\authorrunning{C. Koll et al.}
\institute{Oregon State University, USA
\email{\{kollch,houssam.abbas\}@oregonstate.edu} \and
Michigan State University, USA
\email{\{momtazan,borzoo\}@msu.edu} \and
Corresponding author.}
\maketitle
\begin{abstract}
We present the first decentralized algorithm for detecting predicates over continuous-time signals under partial synchrony.
A distributed cyber-physical system (CPS) consists of a network of agents, each of which measures (or computes) a continuous-time signal.
Examples include distributed industrial controllers connected over wireless networks and connected vehicles in traffic.
The safety requirements of such CPS, expressed as logical predicates, must be monitored at runtime.
This monitoring faces three challenges: first, every agent only knows its own signal, whereas the safety requirement is global and carries over multiple signals.
Second, the agents' local clocks drift from each other, so they do not even agree on the time. Thus, it is not clear which signal values are actually synchronous to evaluate the safety predicate.
Third, CPS signals are continuous-time so there are potentially uncountably many safety violations to be reported.
In this paper, we present the first decentralized algorithm for detecting conjunctive predicates in this setup.
Our algorithm returns all possible violations of the predicate, which is important for eliminating bugs from distributed systems regardless of actual clock drift.
We prove that this detection algorithm is in the same complexity class as the detector for discrete systems.
We implement our detector and validate it experimentally. 
\keywords{Predicate detection  \and Distributed systems \and Partial synchrony \and Cyber-physical systems.}
\end{abstract}

\section{Introduction: Detecting All Errors in Distributed CPS}
\label{sec:intro}

This paper studies the problem of detecting all property violations in a distributed cyber-physical systems (CPS).
A \emph{distributed CPS} consists of a network of communicating agents.
Together, the agents must accomplish a common task and preserve certain properties.
For example, a network of actuators in an industrial control system must maintain a set point, or a swarm of drones must maintain a certain geometric formation.
In these examples, we have $N$ agents generating $N$ continuous-time and real-valued signals $x_n, 1 \leq n \leq N$, and a \emph{global property} of all these signals must be maintained, such as the property $(x_1>0) \land\dots (x_N>0)$.
At runtime, an algorithm continuously monitors whether the property holds.

These systems share the following characteristics:
first, CPS signals are \emph{analog} (continuous-time, real-valued), and the global properties are continuous-time properties.
From a distributed computing perspective, this means that every moment in continuous-time is an event, yielding uncountably many events.
Existing reasoning techniques from the discrete time settings, by contrast, depend on there being at the most countably many events.

Second, each agent in these CPS has a \emph{local clock} that drifts from other agents' clocks: so if agent 1 reports $x_1(3)=5$ and agent 2 reports that $x_2(3)=-10$, these are actually not necessarily synchronous measurements.
So we must re-define property satisfaction to account for unknown drift between clocks.
For example, if the local clocks drift by at most 1 second, then the monitor must actually check whether any value combination of $x_1(t), x_2(s)$ violates the global property, with $|t-s|\leq 1$.
We want to identify any scenario where the system execution \emph{possibly} \cite{cooper1991consistent} violates the global property; the actual unknown execution may or may not do so.

Clock drift raises a third issue: the designers of distributed systems want to know \emph{all the ways} in which an error state could occur.
E.g., suppose again that the clock drift is at most $1$, and the designer observes that the values $(x_1(1), x_2(1.1))$ violate the specification, and eliminates this bug.
But when she reruns the system, the actual drift is 0.15 and the values $(x_1(1), x_2(1.15))$ also violate the spec.
Therefore all errors, resulting from all possible clock drifts within the bound, must be returned to the designers.
This way the designers can guarantee the absence of failures regardless of the actual drift amount.
When the error state is captured in a predicate, this means that all possible satisfactions of the predicate must be returned.
This is known as the \emph{predicate detection} problem.
We distinguish it from predicate monitoring, which requires finding only one such satisfaction, not all.

Finally, in a distributed system, a central monitor which receives all signals is a single point of failure: if the monitor fails, predicate detection fails. Therefore, ideally, the detection would happen in decentralized fashion.

In this work, we solve the problem of decentralized predicate detection for distributed CPS with drifting clocks under partial synchrony.

\paragraph{Related Work}
There is a rich literature dealing with decentralized predicate detection in \emph{the discrete-time setting}: e.g.
\cite{chauhan2013distributed} detects regular discrete-time predicates,
while \cite{garg2020predicate} detects lattice-linear predicates over discrete states,
and \cite{sen2004detecting} performs detection on a regular subset of Computation Tree Logic (CTL).
We refer the reader to the books by Garg \cite{gargbook} and Singhal \cite{singhalbook} for more references.
By contrast, we are concerned with \emph{continuous-time} signals, which have uncountably many events and necessitate new techniques.
For instance, one cannot directly iterate through events as done in the discrete setting.

The recent works \cite{momtaz2021predicate,momtaz23iccps} do \emph{monitoring} of temporal formulas over partially synchronous analog distributed systems -- i.e., they only find one satisfaction, not all.
Moreover, their solution is centralized.

More generally, one finds much work on monitoring temporal logic properties, especially Linear Temporal Logic (LTL) and Metric Temporal Logic (MTL), but they either do monitoring or work in discrete-time.
Notably, \cite{basin2015failure} used a three-valued MTL for monitoring in the presence of failures and non-FIFO communication channels.
\cite{bauer2016decentralised} monitors satisfaction of an LTL formula.
\cite{mostafa2015decentralized} considered a three-valued LTL for distributed systems with asynchronous properties.
\cite{bataineh2019efficient} addressed the problem with a tableau technique for three-valued LTL.
Finally, \cite{sen2004efficient} considered a past-time distributed temporal logic which emphasizes distributed properties over time.

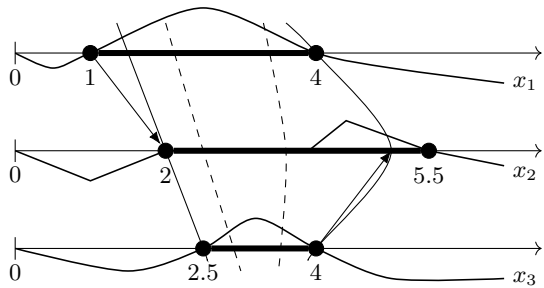
\begin{figure}[t]
	\centering
	\begin{tikzpicture}
[event/.style={circle,draw,fill=white,inner sep=0pt,minimum size=2mm}]
\foreach \n in {1,...,3} {
    \coordinate (process \n) at (1,-1.3*\n);
    \draw[->] (process \n) -- +(7,0);
}

\node[event,fill=black,label=below:$1$] (left root 1) at (2,0 |- process 1) {};
\node[event,fill=black,label=below:$4$] (right root 1) at (5,0 |- process 1) {};
\draw[line width=.08cm] (left root 1) -- (right root 1);
\node[event,fill=black,label=below:$2$] (left root 2) at (3,0 |- process 2) {};
\node[event,fill=black,label=below:$5.5$] (right root 2) at (6.5,0 |- process 2) {};
\draw[line width=.08cm] (left root 2) -- (right root 2);
\node[event,fill=black,label=below:$2.5$] (left root 3) at (3.5,0 |- process 3) {};
\node[event,fill=black,label=below:$4$] (right root 3) at (5,0 |- process 3) {};
\draw[line width=.08cm] (left root 3) -- (right root 3);

\draw (process 1) +(0,0.15) -- +(0,-0.15);
\node at (1,-1.4) [anchor=north] {0};
\draw (process 2) +(0,0.15) -- +(0,-0.15);
\node at (1,-2.7) [anchor=north] {0};
\draw (process 3) +(0,0.15) -- +(0,-0.15);
\node at (1,-4) [anchor=north] {0};

\draw plot [smooth] coordinates {(2.35,-0.9) (2.5,-1.3) (left root 2) (left root 3) (3.55,-4.05)};
\draw plot [smooth] coordinates {(4.6,-0.9) (right root 1) (6,-2.6) (right root 3) (4.9,-4.05)};
\draw[dashed] plot [smooth] coordinates {(3,-0.9) (3.5,-2.6) (4,-4.2)};
\draw[dashed] plot [smooth] coordinates {(4.3,-0.9) (4.6,-2.6) (4.5,-4.2)};

\draw[-{Latex}] (left root 1) -- (left root 2);
\draw[-{Latex}] (right root 3) -- (6, -2.6);

\datavisualization[xy Cartesian, visualize as smooth line]
    data {
        x, y
        1, -1.3
        1.5, -1.5
        2, -1.3
        3.5, -0.7
        5, -1.3
        6, -1.5
        7.5, -1.7
    };
\node at (7.5,-1.7) [anchor=west] {$x_1$};
\datavisualization[xy Cartesian, visualize as line]
    data {
        x, y
        1, -2.6
        2, -3
        3, -2.6
        4.9, -2.6
        5.4, -2.2
        6.5, -2.6
        7.5, -2.8
    };
\node at (7.5,-2.9) [anchor=west] {$x_2$};
\datavisualization[xy Cartesian, visualize as smooth line]
    data {
        x, y
        1, -3.9
        2.5, -4.2
        3.5, -3.9
        4.2, -3.5
        5, -3.9
        6, -4.3
        7.5, -4.3
    };
\node at (7.5,-4.3) [anchor=west] {$x_3$};
\end{tikzpicture}
	\caption{An example of a continuous-time distributed signal with 3 agents. Three timelines are shown, one per agent. The signals $x_n$ are also shown, and the local time intervals over which they are non-negative are solid black. The skew $\cskew$ is 1. The Happened-before relation is illustrated with solid arrows, e.g. between $\evnt{1}{1} \hb \evnt{2}{2}$, and $\evnt{3}{4} \hb \evnt{2}{5}$. These are not message transmission events, rather they follow from \autoref{def:dstr signal}. Some satisfying cuts for the predicate $\phi = (x_1 \geq 0) \land (x_2 \geq 0) \land (x_3 \geq 0)$ are shown as dashed arcs, and the extremal cuts as solid arcs. All extremal cuts contain root events, and leftmost cut $A$ also contains non-root events.}
	\label{fig:continuous satcut}
\end{figure}

\paragraph{Illustrative Example.}
It is helpful to overview our algorithm and key notions via an example before delving into the technical details. An example is shown in \autoref{fig:continuous satcut}. Three agents produce three signals $x_1,x_2,x_3$.
The decentralized detector consists of three local detectors $D_1,D_2,D_3$, one on each agent. Each $x_n$ is observed by the corresponding $D_n$.
The predicate $\phi = (x_1 \geq 0) \land (x_2 \geq 0) \land (x_3 \geq 0)$ is being detected. It is possibly true over the intervals shown with solid black bars; their endpoints are measured on the local clocks.
The detector only knows that the maximal clock skew is $\cskew=1$, but not the actual value, which might be time-varying.

Because of clock skew, any two local times within $\cskew$ of each other must be considered as potentially concurrent, i.e. they might be measured at a truly synchronous moment.
For example, the triple of local times $[4,4.5,3.6]$ might have been measured at the global time $4$, in which case the true skews were 0, 0.5, and -0.4 respectively.
Such a triple is (loosely speaking) called a \emph{consistent cut} (\autoref{def:ccut}).
The detector's task is to find all consistent cuts that satisfy the predicate.
In continuous time, there can be uncountably many, as in \autoref{fig:continuous satcut}; the dashed lines show two satisfying consistent cuts, or satcuts for short.

In this example, our detector outputs two satcuts, $[1.5, 2, 2.5]$ and $[4, 5, 4]$, shown as thin solid lines.
These two have the special property (shown in this paper) that every satcut lies between them, and every cut between them is a satcut.
For this reason we call them \emph{extremal satcuts} (\autoref{def:leftmost cut}).
Thus these two satcuts are a finite representation of the uncountable set of satcuts, and encode all the ways in which the predicate might be satisfied.

We note three further things: the extremal satcuts are not just the endpoints of the intervals, and simply inflating each interval by $\cskew$ and intersecting them does not yield the satcuts.
Each local detector must somehow learn of the relevant events (and only those) on other agents, to determine whether they constitute extremal satcuts.

\paragraph{Contributions}
In this paper, we present the first \emph{decentralized predicate detector for distributed CPS}, thus enhancing the rigor of distributed CPS design.
\begin{itemize}
	\item Our solution is fully decentralized: each agent only ever accesses its own signal, and exchanges a limited amount of information with the other agents.
	\item It is an online algorithm, running simultaneously with the agents' tasks.
	\item It applies to an important class of global properties that are conjunctions of local propositions.
	\item We introduce a new notion of clock, the \emph{physical vector clock}, which might be of independent interest. A physical vector clock orders continuous-time events in a distributed computation without a shared clock.
	\item Our algorithm can be deployed on top of existing infrastructure. Specifically, our algorithm includes a modified version of the classical detector of \cite{chauhan2013distributed}, and so can be deployed on top of existing infrastructure which already supports that detector.
\end{itemize}

\paragraph{Organization}
In \autoref{sec:prelims}, we give necessary definitions and define the problem.
In \autoref{sec:satcuts} we establish fundamental properties of the uncountable set of events $S_E$ satisfying the predicate.
Our detector is made of two processes: a decentralized abstractor presented in \autoref{sec:abstractor}, and a decentralized slicer presented in \autoref{sec:ourslicer}.
Together, they compute a finite representation of the uncountable $S_E$.
The complexity of the algorithm is also analyzed in \autoref{sec:ourslicer}.
\autoref{sec:experiments} demonstrates an implementation of the detector, and \autoref{sec:conclusion} concludes.
All proofs are in the Appendix.

\section{Preliminaries and Problem Definition}
\label{sec:prelims}

We first set some notation. The set of reals is $\reals$, the set of non-negative reals is $\nnreals$.
The integer set $\{1,\dots,N\}$ is abbreviated as $[N]$. \emph{Global} time values
(kept by an \emph{imaginary} global clock) are denoted by $\gclk$, $\gclk'$, $\gclk_1$, $\gclk_2$, etc, while the symbols $t$, $t'$, $t_1$, $t_2$, $s,s',s_1,s_2$, etc. denote \emph{local} clock values specific to given agents which will always be clear from the context.
A \emph{lattice} is a set $S$ equipped with a partial order relation $\sqsubseteq$ s.t. every 2 elements have a supremum, or \emph{join}, and an infimum, or \emph{meet}.
An \emph{increasing} function $f$ is one s.t. $t<t' \implies f(t)<f(t')$.
Notation $(x_n)_n$ indicates a sequence $(x_1,\ldots,x_N)$ where $N$ is always clear from context.

\subsection{The Continuous-Time Setup}
\label{sec:cont setup}
This section defines the setup of this study. It generalizes the classical discrete-time setup, and follows closely the setup in \cite{momtaz2021predicate}.
We assume a loosely coupled system with asynchronous message passing between agent monitors.
Specifically, the system consists of $N$ reliable \emph{agents} that do not fail, denoted by $\{\agent_1, \agent_2, \dots, \agent_N\}$, without any shared memory or global clock.
The output signal of agent $\agent_n$ is denoted by $x_n$, for $1 \leq n \leq N$.
Agents can communicate via FIFO lossless channels.
There are no bounds on message transmission times.

In the discrete-time setting, an event is a value change in an agent's variables.
The following definition generalizes this to the continuous-time setting.

\begin{definition}[Output signal and events]
	\label{def:cont signal}
	An \emph{output signal} (of some agent) is a function $x: \nnreals \rightarrow \reals$, which is right-continuous (i.e., $\lim_{s \rightarrow t^+} x(s) = x(t)$ at every $t$), and left-limited (i.e., $\lim_{s \rightarrow t^-} |x(s)| < \infty$ for all $t$).

	In an agent $\agent_n$, an \emph{event} is a pair $(t,x_n(t))$, where $t$ is the \emph{local} time kept by the agent's local clock.
	This will often be abbreviated as $\evnt{n}{t}$ to follow standard notation from the discrete-time literature.
\end{definition}

Note that an output signal can contain discontinuities.

\begin{definition}[Left and right roots]
	\label{def:roots}
	A \emph{root} is an event $\evnt{n}{t}$ where $x_n(t) = 0$ or a discontinuity at which the signal changes sign: $\text{sgn}( x_n(t)) \neq \text{sgn}(\lim_{s \rightarrow t^-} x_n(s))$.
	A \emph{left root} $\evnt{n}{t}$ is a root preceded by negative values: there exists a positive real $\delta$ s.t. $x_n(t - \alpha) < 0$ for all $0 < \alpha \leq \delta$.
	A \emph{right root} $\evnt{n}{t}$ is a root followed by negative values: $x_n(t + \alpha) < 0$ for all $0 < \alpha \leq \delta$.
\end{definition}
In \autoref{fig:continuous satcut}, the only left root of $x_2$ is $\evnt{2}{2} = (2,x_2(2))=(2,0)$.
The single right root of $x_2$ is $\evnt{2}{5.5}$.
Notice that intervals where the signal is identically 0 are allowed, as in $x_2$.

We will need to refer to a global clock which acts as a `real' time-keeper.
This global clock is a theoretical object used in definitions and theorems, and is \emph{not} available to the agents.
We make these assumptions:
\begin{assumption}
	\label{ass:basic}
	\begin{enumerate}[label=(\alph*)]
	\item (Partial synchrony) The local clock of an agent $\agent_n$ is an increasing function $c_n: \nnreals \rightarrow \nnreals$, where $c_n(\gclk)$ is the value of the local clock at global time $\gclk$.
	For any two agents $\agent_n$ and $\agent_m$, we have:
	$$\forall \gclk \in \nnreals: |c_n(\gclk)-c_{m}(\gclk)| < \cskew$$
	with $\cskew > 0$ being the maximum \emph{clock skew}.
	The value $\cskew$ is known by the detector in the rest of this paper.
	In the sequel, we make it explicit when we refer to `local' or `global' time.\label{ass:partial synchrony}
	\item (Starvation-freedom and non-Zeno) Every signal $x_n$ has infinitely many roots in $\nnreals$, with a finite number of them occurring in any bounded interval.\label{ass:starvation freedom}
	\end{enumerate}
\end{assumption}

\begin{remark}
Our detection algorithm can trivially handle multi-dimensional output signals $x_n$.
We skip this generalization for clarity of exposition.
\end{remark}
\begin{remark}
\noindent In distributed systems, agents typically exchange messages as part of normal operation. These messages help establish an ordering between events (a Send occurs before the corresponding Receive).
This extra order information can be incorporated in our detection algorithm with extra bookkeeping.
\\
We do \emph{not} assume that the clock drift is constant -- it can vary with time. It is assumed to be uniformly bounded by $\cskew$, which can be achieved by using a clock synchronization algorithm, like NTP~\cite{ntp}.
\end{remark}

A distributed signal is modeled as a set of events partially ordered by Lamport's \emph{happened-before} relation~\cite{lamport1978time}, adapted to the continuous-time setting.

\begin{definition}[Analog Distributed signal]
	\label{def:dstr signal}
	A \emph{distributed signal} on $N$ agents is a tuple $(E,\hb)$, in which $E$ is a set of events
	$$E = \{\evnt{n}{t} \such n \in [N] ,  ~t\in \nnreals\}$$
	such that for all $t \in \nnreals, n \in [N]$, there exists an event $\evnt{n}{t}$ in $E$, and $t$ is local time in agent $\agent_n$.
	The \emph{happened-before} relation $\hb \subseteq E\times E $ between events is such that:
	\begin{enumerate}[label=(\arabic*)]
		\item\label{def:hb local} In every agent $\agent_n$, all events are totally ordered, that is,
		$$\forall t, t' \in \nnreals: (t < t') \implies (\evnt{n}{t} \hb \evnt{n}{t'}).$$
		\item\label{def:hb eps} For any two events $\evnt{n}{t}, \evnt{m}{t'} \in E$, if $t +\cskew \leq t'$, then $\evnt{n}{t} \hb \evnt{m}{t'}$.
		\item\label{def:hb transitive} If $e \hb f$ and $f \hb g$, then $e \hb g$.
	\end{enumerate}
	We denote $E[n]$ the subset of events that occur on $A_n$, i.e. $E[n] \defeq \{\evnt{n}{t} \in E\}$.
\end{definition}

The happened-before relation, $\hb$, captures what can be known about event ordering in the absence of perfect synchrony.
Namely, events on the same agent can be linearly ordered, and at least an $\cskew$ of time must elapse between events on different agents for us to say that one happened before the other.
Events from different agents closer than an $\cskew$ apart are said to be \emph{concurrent}.

\paragraph{Conjunctive Predicates}
This paper focuses on specifications expressible as \emph{conjunctive predicates} $\formula$, which are conjunctions of $N$ linear inequalities.
\begin{equation}
\label{eq:conjunctive predicate}
	\formula \defeq (x_1\geq 0)\land (x_2 \geq 0) \land \ldots \land (x_N \geq 0).
\end{equation}
These predicates model the simultaneous co-occurrence, in global time, of events of interest, like `all drones are dangerously close to each other'.
Eq. \eqref{eq:conjunctive predicate} also captures the cases where some conjuncts are of the form $x_n\leq 0$ and $x_n=0$.
If $N$ numbers $(a_n)$ satisfy predicate $\formula$ (i.e., are all non-negative), we write this as $(a_1,\ldots, a_N)\models \formula$.
Henceforth, we say `predicate' to mean a conjunctive predicate.
Note that the restriction to linear inequalities does not significantly limit our ability to model specifications.
If an agent $n$ has some signal $x_n$ with which we want to check $f(x_n) \geq 0$ for some arbitrary function $f$, then the agent can generate an auxiliary signal $y_n := f(x_n)$ so that we can consider the linear inequality $y_n \geq 0$.

What does it mean to say that a distributed signal satisfies $\formula$? And at what moment in time?
In the ideal case of perfect synchrony ($\cskew=0$) we'd simply say that $E$ satisfies $\formula$ at $\chi$ whenever $(x_1(\chi),\ldots,x_N(\chi))\models \formula$.
We call such a synchronous tuple $(x_n(\chi))_{n}$ a \emph{global state}.
But because the agents are only synchronized to within an $\cskew>0$, it is not possible to guarantee evaluation of the predicate at true global states.
The conservative thing is to treat concurrent events, whose local times differ by less than $\cskew$, as being simultaneous on the global clock.
E.g., if $N=2$ and $\cskew=1$ then $(x_1(1),x_2(1.5))$ is treated as a possible global state.
The notion of consistent cut, adopted from discrete-time distributed systems~\cite{garg01slicing}, formalizes this intuition.

\begin{definition}[Consistent Cut]
	\label{def:ccut}
	Given a distributed signal $(E, \hb)$, a subset of events $C \subset E$ is said to form a \emph{consistent cut} if and only if when $C$ contains an event $e$, then it contains all events that happened-before $e$.
	Formally,
	\begin{equation}
		\label{eq:cc hb def}
		\forall e \in E: (e \in C) \, \wedge \, (f \hb e) \implies f \in C.
	\end{equation}
	We write $C[n]$ for the cut's local events produced on $A_n$, and $C^\tau[n] \defeq \{t\such \evnt{n}{t}\in C[n]\}$ for the timestamps of a cut's local events.
\end{definition}
\noindent From this and \autoref{def:dstr signal}~\ref{def:hb transitive} it follows that if $\evnt{m}{t'}$ is in $C$, then $C$ also contains every event $\evnt{n}{t}$ such that $t+\cskew \leq t'$.
Thus to avoid trivialities, we may assume that $C$ contains at least one event from every agent.

A consistent cut $C$ is represented by its \emph{frontier} $\front(C) = \left(\evnt{1}{t_1},\ldots, \evnt{N}{t_N}\right)$, in which each $\evnt{n}{t_n}$ is the last event of agent $\agent_n$ appearing in $C$.
Formally:
$$\forall n \in [N],~ t_n \defeq \sup C^\tau[n] = \sup \{t \in \nnreals \such ~\evnt{n}{t} \in C[n]\}.$$

Henceforth, we simply say `cut' to mean a consistent cut, and we denote a frontier by $(\evnt{n}{t_n})_n$.
We highlight some easy yet important consequences of the definition: on a given agent $A_n$, $\evnt{n}{t} \in C$ for all $t < t_n$, so the timestamps of the cut's local events, $C^\tau[n]$, form a left-closed interval of the form $[0,a]$, $[0,a)$ or $[0,\infty)$.
Moreover, either $C^\tau[n]=[0,\infty)$ for all $n$, in which case $C=E$, or every $C^\tau[n]$ is bounded, in which case every $t_n$ is finite and $|t_n-t_m|\leq \cskew$ for all $n,m$.
\emph{Thus the frontier of a cut is a possible global state}.
This then justifies the following definition of distributed satisfaction.

\begin{definition}[Distributed Satisfaction; $S_E$]
	\label{def:satcut}
	Given a predicate $\formula$, a distributed signal $(E,\hb)$ over $N$ agents, and a consistent cut $C$ of $E$ with frontier
	$$\front(C) = \Big(\,(t_1,x_1(t_1)),\ldots,(t_N,x_N(t_N))\,\Big)$$
	we say that $C$ \emph{satisfies} $\formula$ iff $\big(x_1({t_1}), x_2(t_2),\ldots, x_N(t_N)\big)\models \formula$.
	We write this as $C \models \formula$, and say that $C$ is a \emph{satcut}.
	The set of all satcuts in $E$ is written $S_E$.
\end{definition}

\subsection{Problem Definition: Decentralized Predicate Detection}
The detector seeks to find \emph{all} possible global states that satisfy a given predicate, i.e. all satcuts in $S_E$. In general, $S_E$ is uncountable.

\textbf{Architecture.}
The system consists of $N$ agents with partially synchronous clocks with drift bounded by a known $\cskew$, generating a continuous-time distributed signal $(E,\hb)$.
Agents communicate in a FIFO manner, where messages sent from an agent $A_1$ to an agent $A_2$ are received in the order that they were sent.

\textbf{Problem statement.}
Given $(E,\hb)$ and a conjunctive predicate $\formula$, find a decentralized detection algorithm that computes a finite representation of $S_E$.
The detector is decentralized, meaning that it consists of $N$ local detectors, one on each agent, with access only to the local signal $x_n$ (measured against the local clock), and to messages received from other agents' detectors.

By computing a representation of all of $S_E$ (and not some subset), we account for asynchrony and the unknown orderings of events within $\cskew$ of each other.
One might be tempted to propose something like the following algorithm: detect all roots on all agents, then see if any $N$ of them are within $\cskew$ of each other.
This quickly runs into difficulties: first, a satisfying cut is not necessarily made up of roots; some or all of its events can be interior to the intervals where $x_n$'s are positive (see \autoref{fig:satcut lattice}).
Second, the relation between roots and satcuts must be established: it is not clear, for example, whether even satcuts made of only roots are enough to characterize all satcuts (it turns out, they're not).
Third, we must carefully control how much information is shared between agents, to avoid the detector degenerating into a centralized solution where everyone shares everything with everyone else.

\section{The Structure of Satisfying Cuts}
\label{sec:satcuts}

We establish fundamental properties of satcuts. In the rest of this paper we exclude the trivial case $C=E$.
Proposition \ref{thm:cont lattice} mirrors a discrete-time result~\cite{chauhan2013distributed}.

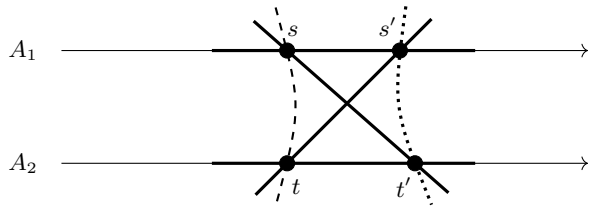
\begin{figure}[t]
	\centering
	\begin{tikzpicture}
[event/.style={circle,draw,fill=black,inner sep=0pt,minimum size=2mm}]
\coordinate (process 1) at (0,1.5);
\draw[->] (process 1) -- +(7,0);
\node at (-0.2,1.5) [anchor=east] {$A_1$};
\coordinate (process 2) at (0,0);
\draw[->] (process 2) -- +(7,0);
\node at (-0.2,0) [anchor=east] {$A_2$};

\node[event,label={[above]45:$s$}] (s) at (3,0 |- process 1) {};
\node[event,label={[above left]45:$s'$}] (s') at (4.5,0 |- process 1) {};

\node[event,label={[below right]-120:$t$}] (t) at (3,0 |- process 2) {};
\node[event,label={[below left]-45:$t'$}] (t') at (4.7,0 |- process 2) {};

\draw[very thick,shorten >=-.7cm,shorten <=-.7cm] (s) -- (t');
\draw[very thick,shorten >=-.7cm,shorten <=-.7cm] (t) -- (s');

\draw[very thick] (2,0 |- process 1) -- (5.5,0 |- process 1);
\draw[very thick] (2,0 |- process 2) -- (5.5,0 |- process 2);

\draw[dashed,thick,shorten >=-.7cm,shorten <=-.7cm] (s) .. controls (3.2,0.75) .. (t);
\draw[dotted,very thick,shorten >=-.7cm,shorten <=-.7cm] (s') .. controls (4.4,0.75) .. (t');
\end{tikzpicture}
	\caption{Two satcuts for a pair of agents $A_1$ and $A_2$, shown by the solid lines $(s, t')$ and $(s', t)$. Their intersection is $(s, t)$, shown by a dashed arc, and their union is $(s', t')$, shown by a dotted arc. For a conjunctive predicate $\formula$, the intersection and union are also satcuts, forming a lattice of satcuts.}
	\label{fig:satcut lattice}
\end{figure}

\begin{proposition}
	\label{thm:cont lattice}
	The set of satcuts for a conjunctive predicate is a lattice where the join and meet are the union and intersection operations, respectively.
\end{proposition}

We show that the set of satcuts is characterized by special elements, which we call the leftmost and rightmost cuts.
\begin{definition}[Extremal cuts]
	\label{def:leftmost cut}
	Let $S_E$ be the \emph{set} of all satcuts in a given distributed signal $(E,\hb)$.
	For an arbitrary $C \in S_E$ with frontier $(\evnt{n}{t_n})_n$ and positive real $\alpha$, define $C-\alpha$ to be the set of cuts whose frontiers are given by
	$$(\evnt{1}{t_1 - \delta_1}, \evnt{2}{t_2 - \delta_2}, \ldots, \evnt{N}{t_N - \delta_N}) \textrm{  s.t. for all }~n:0\leq \delta_n\leq \alpha \textrm{ and for some } n.~\delta_n>0$$
	A {\em leftmost satcut} is a satcut $C \in S_E$ for which there exists a positive real $\alpha$ s.t. $C-\alpha$ and $S_E$ do not intersect.
	The set $C+\alpha$ is similarly defined.
	A \emph{rightmost cut} $C$ (not necessarily sat) is one for which there exists a positive real $\alpha$ s.t. $C+\alpha$ and $S_E$ do not intersect, and $C-\alpha \subset S_E$.
	We refer to leftmost and rightmost (sat)cuts as \emph{extremal cuts}.
\end{definition}
Intuitively, $C-\alpha$ ($C+\alpha)$ is the set of all cuts one obtains by slightly moving the frontier of $C$ to the left (right) by amounts less than $\alpha$.
If doing so always yields non-satisfying cuts, then $C$ is a leftmost satcut.
If moving $C$ slightly to the right always yields unsatisfying cuts, but moving it slightly left yields satcuts, then $C$ is a rightmost cut.
The reason we don't speak of rightmost \emph{sat}cuts is that we only require signals to be left-limited, not continuous.
If signals $x_n$ are all continuous, then rightmost cuts are all satisfying as well.

In a signal, there are multiple extremal cuts.
\autoref{fig:satcut lattice} suggests, and Lemma~\ref{thm:satcut} proves, that all satcuts live between a leftmost satcut and rightmost cut.

\begin{lemma}[Satcut intervals]
	\label{thm:satcut}
	Every satcut of a conjunctive predicate lies in-between a leftmost satcut and rightmost cut, and there are no non-satisfying cuts between a leftmost satcut and the first rightmost cut that is greater than it in the lattice order.
\end{lemma}

Thus we may visualize satcuts as forming $N$-dimensional intervals with endpoints given by the extremal cuts.
The main result of this section states that there are finitely many extremal satcuts in any bounded time interval, so the extremal satcuts are the finite representation we seek for $S_E$.

\begin{theorem}
\label{thm:finitely many satcuts}
	A distributed signal has finitely many extremal satcuts in any bounded time interval.
\end{theorem}
Therefore, it is conceivably possible to recover algorithmically the extremal cuts, and therefore all satcuts by \autoref{thm:satcut}.
The rest of this paper shows how.

\section{The Abstractor Process}
\label{sec:abstractor}

Having captured the structure of satcuts, we now define the distributed \emph{abstractor process} that will turn our continuous-time problem into a discrete-time one, amenable to further processing by our modified version of the slicer algorithm of \cite{chauhan2013distributed}.
This abstractor also has the task of creating a happened-before relation.
We first note a few complicating factors.
First, this will not simply be a matter of sampling the roots of each signal.
That is because extremal cuts can contain non-root events, as shown in \autoref{fig:continuous satcut}.
Thus the abstractor must somehow find and sample these non-root events as part of its operation.
Second, as in the discrete case, we need a kind of clock that allows the local detector to know the happened-before relation between events.
The local timestamp of an event, and existing clock notions, are not adequate for this.
Third, to establish the happened-before relation, there is a need to exchange event information between the processes, without degenerating everything into a centralized process (by sharing everything with everyone).
This complicates the operation of the local abstractors, but allows us to cut the number of messages in half.

\subsection{Physical Vector Clocks}
\label{sec:pvcs}

We first define Physical Vector Clocks (PVCs), which generalize vector clocks~\cite{mattern1989virtual} from countable to uncountable sets of events.
They are used by the abstractor process (next section) to track the happened-before relation.
A PVC captures one agent's knowledge, at appropriate local times, of events at other agents.
\begin{definition}[Physical Vector Clock]
	\label{def:pvc}
	Given a distributed signal $(E,\hb)$ on $N$ agents, a \emph{Physical Vector Clock}, or PVC, is a set of $N$-dimensional timestamp vectors $\pvc{n}{t} \in \nnreals^N$, where vector $\pvc{n}{t}$ is defined by the following:
	\begin{enumerate}[label=(\arabic*)]
		\item\label{def:pvc init} Initialization:
		$\pvc{n}{0}[i] = 0,\quad \forall i \in \{1, \dots, N\}$
		\item\label{def:pvc local} Timestamps store the local time of their agent:
		$\pvc{n}{t}[n] = t$ for all $t>0$.
		\item\label{def:pvc hb} Timestamps keep a consistent view of time: Let $V_n^t$ be the set of all timestamps $\pvc{m}{s}$ s.t. $\evnt{m}{s}$ happened-before $\evnt{n}{t}$ in $E$.
		Then:
		\[\pvc{n}{t}[i] = \max_{\pvc{m}{s} \in V_n^t}(\pvc{m}{s}[i]),\quad \forall i \in [N]\setminus\{n\}, t>0\]
	\end{enumerate}
PVCs are partially ordered: $\pvc{n}{t} < \pvc{m}{t'}$ iff $\pvc{n}{t} \neq \pvc{m}{t'}$ and $\pvc{n}{t}[i] \leq \pvc{m}{t'}[i]~\forall i \in [N]$.
\end{definition}
We say $\pvc{n}{t}$ is \emph{assigned} to $\evnt{n}{t}$.
The detection algorithm can now know the happened-before relation by comparing PVCs.

\begin{theorem}
	\label{thm:pvc isomorphic}
	Given a distributed signal $(E,\hb)$, let $V$ be the corresponding set of PVC timestamps.
	Then $(V,<)$ and $(E,\hb)$ are order isomorphic, i.e., there is a bijective mapping between $V$ and $E$ s.t. $\evnt{n}{t} \hb \evnt{m}{t'}$ iff $\pvc{n}{t} < \pvc{m}{t'}$.
\end{theorem}

\autoref{def:pvc} is not quite a constructive definition. We need a way to actually compute PVCs.
This is enabled by the next theorem.

\begin{theorem}
	\label{thm:pvc assign}
	The assignment
	\[\pvc{n}{t} = \left\lbrace \begin{matrix} [0,\ldots,0,t,0,\ldots,0], & t<\cskew \\
	[t - \cskew, \dots, t - \cskew, t, t - \cskew, \dots, t - \cskew], & t \geq \cskew\\
	\end{matrix}\right. \]
	where the $t$ is in the $n^{th}$ position in both cases, satisfies the conditions of PVC in \autoref{def:pvc}.
\end{theorem}

\subsection{Abstractor Description}
\label{sec:abstractor description}

\begin{figure}[t]
	\centering
	\scalebox{0.9}{\begin{tikzpicture}
[event/.style={circle,draw,fill=white,inner sep=0pt,minimum size=2mm}]
\coordinate (process 1 sig) at (0,5);
\draw[->] (process 1 sig) -- +(7,0);
\node at (-0.2,5) [anchor=east] {$A_1$};
\coordinate (process 2 sig) at (0,3.5);
\draw[->] (process 2 sig) -- +(7,0);
\node at (-0.2,3.5) [anchor=east] {$A_2$};
\datavisualization[xy Cartesian, visualize as smooth line]
    data {
        x, y
        0.5, 4.5
        1.5, 5
        3.5, 5.5
        5.5, 5
        6, 4.5
        6.8, 4.5
    };
\datavisualization[xy Cartesian, visualize as smooth line]
    data {
        x, y
        0.5, 3
        2, 3
        3, 3.5
        4, 4
        5.3, 3.5
        6.8, 3
    };

\draw[-Implies,double] (3.5,2.5) -- +(0,-0.3);

\coordinate (process 1) at (0,1.5);
\draw[->] (process 1) -- +(7,0);
\node at (-0.2,1.5) [anchor=east] {$A_1$};
\coordinate (process 2) at (0,0);
\draw[->] (process 2) -- +(7,0);
\node at (-0.2,0) [anchor=east] {$A_2$};

\node[event,fill=black] (left root 1) at (1.5,0 |- process 1) {};
\node[event,fill=black] (right root 1) at (5.5,0 |- process 1) {};
\draw[line width=.08cm] (left root 1) -- (right root 1);
\node[event,fill=black] (left root 2) at (3,0 |- process 2) {};
\node[event,fill=black] (right root 2) at (5.3,0 |- process 2) {};
\draw[line width=.08cm] (left root 2) -- (right root 2);

\node[event,fill=black] (nonroot 1) at (6.1,0 |- process 1) {};
\node[event,fill=black] (nonroot 2) at (6.3,0 |- process 2) {};

\draw[-{Latex}] (right root 2) -- (nonroot 1);
\draw[-{Latex}] (right root 1) -- (nonroot 2);

\draw[dashed,-{Latex}] (1.5,5) -- (left root 1);
\draw[dashed,-{Latex}] (3,3.5) -- (left root 2);
\draw[dashed,-{Latex}] (5.5,5) -- (right root 1);
\draw[dashed,-{Latex}] (5.3,3.5) -- (right root 2);
\end{tikzpicture}}
	\caption{A distributed signal of two agents (top) and the output of the abstractor (bottom). The abstractor marks zero-crossings as discrete root events and creates new events (dark circles) to maintain consistency. }
	\label{fig:abstractor work}
\end{figure}
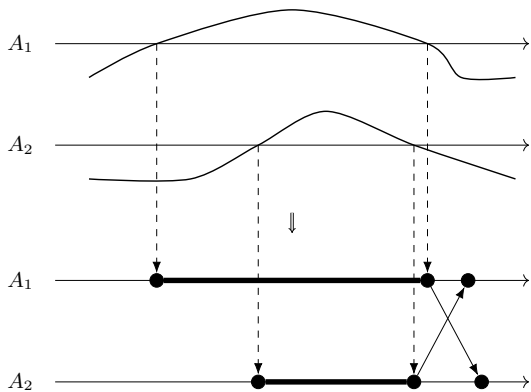

The abstractor is described in \autoref{algo:abstractor} on page \pageref{algo:abstractor}.
Its output is a stream of discrete-time events, their correct PVC values, and the relation $\hb$ between them - i.e., a discrete-time distributed signal.
This signal is processed by the local slicer processes as it is being produced by the abstractor.

\begin{algorithm}[t]
	\DontPrintSemicolon
	\SetAlgoVlined
	\SetStartEndCondition{ }{}{}
	\SetKwProg{Trigger}{trigger}{:}{end}
	\SetKwFor{ForEach}{for each}{:}{end}
	\SetKwIF{If}{ElseIf}{Else}{if}{:}{else if}{else}{end}
	\SetKw{Create}{create}
	\SetKw{Add}{add}
	\SetKw{Send}{send}
	\KwData{Signal of agent $A_n$}
	\KwResult{A stream of discrete events which are roots or $\cskew$-offset from roots}

	\BlankLine
	\Trigger{found a \emph{root} $\evnt{n}{t}$ at local time $t$}{ \label{line:abs:trigger 1}
		\Add{$\evnt{n}{t}$ info} ($n$, $t$, PVC, left or right root) to local buffer\;\label{line:abs:add root}
		\If{$\evnt{n}{t}$ is right root}{
			\ForEach{agent $m \ne n$}{
				\Send{$\evnt{n}{t}$ info} to agent $m$\;
			}
		}
	}
	\BlankLine
	\Trigger{received message about \emph{right root} $\evnt{m}{t}$ from agent $A_m$}{\label{line:abs:trigger 2}
		Set $t' := t + \cskew$, where $\cskew$ is the maximum clock skew\;
		\Create{local event $\evnt{n}{t'}$} \;\label{line:abs:virtual event}
		\Create{relation $\evnt{m}{t} \rightsquigarrow \evnt{n}{t'}$}  (setting the PVC for $\evnt{n}{t'}$ appropriately)\; \label{line:abs:squig}
		\tcc{Info for created event includes that it \emph{came} from a right root $\evnt{m}{t}$, not necessarily that it \emph{is} a root}
		\Add{$\evnt{n}{t'}$ info} ($n$, $t'$, PVC, from right root) to local buffer\;\label{line:abs:add virtual}
		\tcc{Ready events are those whose PVCs will not be updated anymore. See text for details.}
		\If{$A_n$ received at least one message about a \emph{right root} $\evnt{k}{t_k}$ from every other agent $A_k$ such that $t_k \geq t$}{\label{line:abs:ready to send}
			\tcc{Visit events in the buffer, forwarding ones that are ready to the slicer.}
			\ForEach{event $\evnt{n}{s}$ in the local buffer}{
				Set $\pvc{n}{s}[n]=s$ and $\pvc{n}{s}[k]=s-\cskew$ for all $k \neq n$ \;
				Remove $\evnt{n}{s}$ from buffer and \Send{} it to local slicer\;\label{line:abs:remove event}
			}
		}
	}
	\caption{Local abstractor for agent $A_n$}
	\label{algo:abstractor}
\end{algorithm}

The abstractor runs as follows.
It is decentralized, meaning that there is a \emph{local} abstractor running on each agent.
Agent $A_n$'s local abstractor maintains a buffer of discrete events, and consists of two trigger processes.
The first is triggered when a root is detected (by a local zero-finding algorithm; line \ref{line:abs:trigger 1}).
It stores the root's information in a local buffer (for future processing).
\emph{If it is a right root}, it also sends it to the other agents.
The second trigger process (line \ref{line:abs:trigger 2}) is triggered when the agent \emph{receives} a right root information from some other process, at which point it does three things: it creates a local discrete event and a corresponding relation $\rightsquigarrow$ between events (Lines \ref{line:abs:virtual event}-\ref{line:abs:squig}), it updates events in its local buffer to see which ones can be sent to the local slicer process (described in \autoref{sec:ourslicer}), and then it sends them.
It is clear, by construction, that $\rightsquigarrow$ is a happened-before relation: it is the subset of $\hb$ needed for detection purposes.

Before an event $\evnt{n}{t}$ is sent to the slicer, it must have a PVC that correctly reflects the happened-before relation.
This means that all events that happened-before $\evnt{n}{t}$ must be known to agent $n$, which uses them to update the PVC timestamps.
This happens when events have reached agent $A_n$ from every other agent, with timestamps that place them after $\evnt{n}{t}$ (line \ref{line:abs:ready to send}).
This is guaranteed to happen by the starvation-free assumption \ref{ass:basic}.\ref{ass:starvation freedom}.

The output of a local abstractor is a stream of discrete events, so that \emph{the output of the decentralized abstractor as a whole is a distributed discrete-time signal.}
See \autoref{fig:abstractor work}.

Given that all right roots are assigned discrete events by the first trigger, and given that $\cskew$-offset events are also created from them by the second trigger (line \ref{line:abs:virtual event}), we have the following.
\begin{theorem}
	\label{thm:satcuts are same in cont and discrete}
	All events in rightmost cuts are generated by the abstractor.
	Moreover, a rightmost cut of $E$ is also a cut of the discrete signal returned by the abstractor.
\end{theorem}

Thus the slicer process, described in the following section, can find the rightmost cuts when it processes the discrete signal.
What about the leftmost satcuts?
These will be handled by the slicer using the PVCs, as will be shown in the next section.
Doing it this way relieves the abstractor from having to communicate the left roots between processes, thus saving on messages and their wait times.

\section{The Slicer Process for Detecting Predicates}
\label{sec:ourslicer}

The second process in our detector is a decentralized \emph{slicer process}, so-called to keep with the common terminology in discrete distributed systems~\cite{gargbook}.
The slicer is decentralized: it consists of $N$ local slicers $\ourslicer_n$, one per agent.
The slicer runs in parallel with the abstractor and processes the abstractor's output as it is produced.
Recall that the abstractor's output consists of a stream of discrete events, coming from the $N$ agents.
These events are either roots or $\cskew$-offset from roots.
If an event is a left root or $\cskew$-offset from a left root, we will call it a left event.
We define right events similarly.
We will write $F_n$ for those events, output by the abstractor, that occurred on $A_n$.

Every slicer $\ourslicer_n$ maintains a \emph{token} $T_n$, which is a constant-size data structure to keep track of satcuts that contain $A_n$ events.
Specifically, for every event $\evnt{n}{t}$ in $F_n$, the token $T_n$ is forwarded between the agents, collecting information to determine whether there exists a satcut that contains $\evnt{n}{t}$.
We say the slicer is trying to \emph{complete} $\evnt{n}{t}$.
The token's updates are such that it will find that satcut if it exists, or determines that none exists; either way, it is then reset and sent back to its parent process $A_n$ to handle the next event in $F_n$.

Let $\evnt{n}{t}$ be an event that the slicer is currently trying to complete.
The token's updates vary, depending on whether it is currently completing a left event, or a right event.
\emph{If $T_n$ is completing a right event}, the token is updated as follows.
The token currently has a cut whose frontier contains $\evnt{n}{t}$, which is either a satcut or not.
If it is, the token has successfully completed the event and is returned to $A_n$ to handle the next event in $F_n$.
If not, then by the property of \emph{regular} predicates~\cite{chauhan2013distributed}, there exists a \emph{forbidden event} $\evnt{m}{s}$ on the frontier of the cut which either prevents the cut from being consistent or from satisfying the predicate. $T_n$ is sent to the process $A_m$ containing this forbidden event.
$T_n$'s so-called target event, whose inclusion may give $T_n$ a satcut, is the event on $A_m$ following the forbidden $\evnt{m}{s}$.
If the token does not find a next event following $\evnt{m}{s}$, then the token is kept by $\ourslicer_m$ until it receives the next event from the abstractor (which is guaranteed to happen under the starvation-free assumption).
After the token retrieves the next event, the updates to the token and progression of $\ourslicer_n$ then follow the \slicer{}~\cite{chauhan2013distributed}.
Space limitations make it impossible to describe the \slicer{} here, and we refer the reader to the detailed description in \cite{chauhan2013distributed}.

\emph{If handling a left event}, the token is updated as follows.
First, as before, $T_n$ is sent to the process $A_m$ which generates the forbidden $\evnt{m}{s}$ -- i.e., which prevents $T_n$ from completing $\evnt{n}{t}$.
$T_n$'s target event may not be the next event on that process following $\evnt{m}{s}$: that's because if $\evnt{n}{t}$ is a left root, there may exist a left event $\evnt{m}{t - \cskew}$ on $A_m$ which is part of a continuous-time leftmost satcut (by \autoref{def:dstr signal}), but which was not created by the abstractor.
In this case, if the token were to follow the updates for a right event, it would skip a potential satcut.
Instead, the slicer $\ourslicer_m$ will create this event: namely,
if $\ourslicer_m$ sees a new event $\evnt{m}{s'}$ where $s' > t - \cskew$, it knows that $\evnt{m}{t - \cskew}$ has not and will not show up (will not be produced by the abstractor) because  messages are FIFO.
The slicer at this point creates the new event $\evnt{m}{t - \cskew}$.
This is valid since in continuous-time, by definition, every moment has a corresponding event on every agent.
Once the token retrieves this created $\evnt{m}{t - \cskew}$ as its new target, the updates to the token and progression of $\ourslicer_n$ follow the \slicer{}~\cite{chauhan2013distributed}, similarly to the right event scenario.

\underline{Correctness of $\ourslicer$.}
We will show that all extremal cuts of the continuous-time signal are included in the discrete lattice of satcuts of the discrete signal.
Since the \slicer{} computes the discrete lattice, this means in particular that it computes the extremal cuts that are in it.
From these extremal cuts, we can then recover the continuous-time satcuts by \autoref{thm:satcut}.

\begin{theorem}
	\label{thm: mod cgnm correctness}
	Our slicer returns all extremal cuts.
\end{theorem}

We give the space and time complexity of the overall detector.
Since this is an online detector which runs forever (as long as the system is alive), we must fix a time interval for the analysis.
\begin{theorem}
	\label{thm:time complexity}
	The time complexity for each agent is $O(2RN)$, where $R$ is the number of right roots in the given analysis interval.
	The detector consumes $O(N^3)$ memory to store the tokens.
	If roots are uniformly distributed, then the local buffers of the abstractor and slicer grow at the most to size $O(N^2)$.
\end{theorem}
Finally, there is no bound on detection delay, since we don't assume any bounds on message transmission time.
Assuming some bound on transmission delay yields a corresponding bound on detection delay.

\subsection{Worked-out example}
\label{sec:worked-out example}

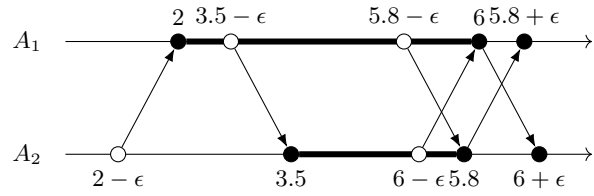
\begin{figure}[t]
	\centering
	\begin{tikzpicture}
[event/.style={circle,draw,fill=white,inner sep=0pt,minimum size=2mm}]
\coordinate (process 1) at (0,1.5);
\draw[->] (process 1) -- +(7,0);
\node at (-0.2,1.5) [anchor=east] {$A_1$};
\coordinate (process 2) at (0,0);
\draw[->] (process 2) -- +(7,0);
\node at (-0.2,0) [anchor=east] {$A_2$};

\node[event,fill=black,label=above:$2$] (left root 1) at (1.5,0 |- process 1) {};
\node[event,fill=black,label=above:$6$] (right root 1) at (5.5,0 |- process 1) {};
\draw[line width=.08cm] (left root 1) -- (right root 1);
\node[event,fill=black,label=below:$3.5$] (left root 2) at (3,0 |- process 2) {};
\node[event,fill=black,label=below:$5.8$] (right root 2) at (5.3,0 |- process 2) {};
\draw[line width=.08cm] (left root 2) -- (right root 2);

\node[event,label=above:$3.5 - \cskew$] (vevent 1) at (2.2,0 |- process 1) {};
\node[event,label=above:$5.8 - \cskew$] (vevent 2) at (4.5,0 |- process 1) {};
\node[event,label=below:$6 - \cskew$] (vevent 3) at (4.7,0 |- process 2) {};
\node[event,label=below:$2 - \cskew$] (vevent 4) at (0.7,0 |- process 2) {};
\node[event,fill=black,label=above:$5.8 + \cskew$] (nonroot 1) at (6.1,0 |- process 1) {};
\node[event,fill=black,label=below:$6 + \cskew$] (nonroot 2) at (6.3,0 |- process 2) {};

\draw[-{Latex}] (vevent 4) -- (left root 1);
\draw[-{Latex}] (vevent 1) -- (left root 2);
\draw[-{Latex}] (vevent 2) -- (right root 2);
\draw[-{Latex}] (vevent 3) -- (right root 1);
\draw[-{Latex}] (right root 2) -- (nonroot 1);
\draw[-{Latex}] (right root 1) -- (nonroot 2);
\end{tikzpicture}
	\caption{Example of \autoref{sec:worked-out example}. Bold intervals are where the local signals are non-negative. The happened-before relation is illustrated with solid arrows. The predicate is $\phi = (x_1 \geq 0) \land (x_2 \geq 0)$. Solid circles represent discrete events returned by the abstractor; hollow circles are those created by the slicers. The leftmost satcut of this example is $[3.5 - \cskew, 3.5]$ and the rightmost is $[6, 5.8]$.}
	\label{fig:worked example}
\end{figure}

We now work through an example execution of the detector on \autoref{fig:worked example}.
We focus on agent $A_2$, its abstractor $\ourabstractor_2$, slicer $\ourslicer_2$ and its token $T_2$.
\begin{enumerate}
	\item Agent $A_2$ encounters a left root in the signal at local time $3.5$.
	This information is forwarded to the abstractor.
	\item The abstractor $\ourabstractor_2$ adds the new root to its buffer with a PVC =$[3.5 - \cskew, 3.5]$.
	\item $A_2$ finds a right root in the signal at local time $5.8$ and forwards it to $\ourabstractor_2$.
	\item The abstractor sends the root information to agent $A_1$.
	It then adds this root to its buffer with a PVC timestamp of $[5.8 - \cskew, 5.8]$.
	\item Abstractor $\ourabstractor_2$ receives a message from $A_1$ about a right root at $A_1$'s local time $6$.
	Note that this is the first knowledge $A_2$ has about anything that is occurring on $A_1$, even though $A_1$ has already found a left root.
	\item $\ourabstractor_2$ uses $A_1$'s message to create a new local event at $6 + \cskew$ with PVC $[6, 6 + \cskew]$.
	\item $\ourabstractor_2$ also adds this new local event to its buffer.
	Since all messages are FIFO, $A_2$ knows that there will be no new messages which will create events before $6 + \cskew$.
	Thus, it can remove both of the events $3.5$ and $5.8$ from the buffer and forward them to its local slicer $\ourslicer_2$.
	At this point both of $A_1$'s events have been forwarded to {\em its} slicer, although $A_2$ has no knowledge of this.
	\item The slicer $\ourslicer_2$ receives an event with a PVC $[3.5 - \cskew, 3.5]$.
	Token $T_2$ is waiting for the next event, so it adds this event to its potential cut.
	\item The token is processed with its new potential cut.
	The cut is found to be inconsistent since $T_2$ has no information about any $A_1$ events.
	\item The token's target is set to be $3.5 - \cskew$ on $A_1$ and the token is sent to $A_1$.
	\item $A_1$ receives $T_2$.
	It walks through its local events $2$ and $6$ and determines that $T_2$'s target event is between the two.
	\item $\ourslicer_1$ creates a new event $\evnt{1}{3.5 - \cskew}$ and notes that $x_1(3.5-\cskew)\geq 0$.
	\item Token $T_2$ incorporates the new event to its potential cut.
	The new potential cut is consistent and satisfies the predicate.
	It is then sent back to $A_2$.
	\item $A_2$ receives $T_2$.
	$T_2$ indicates a satisfying cut, which the agent outputs as a result.
	It then advances $T_2$ to its next event at time $5.8$.
	\item $T_2$ has the current cut of $[3.5 - \cskew, 5.8]$.
	This is not consistent, so it is given the target $5.8 - \cskew$ on $A_1$.
	It is then sent to $A_1$.
	\item $A_1$ receives the token.
	$\ourslicer_1$ walks through its local events and finds that the token's target is between the left root and the right root.
	\item $\ourslicer_1$ creates a new event at $5.8 - \cskew$ and notes that $x_1(5.8 - \cskew)\geq 0$.
	\item The token adds the event to its potential cut.
	It finds that its new potential cut is consistent and satisfies the predicate.
	It is then sent back to $A_2$.
	\item $A_2$ receives $T_2$ and outputs the satcut.
	The algorithm then continues with new events as they occur.
\end{enumerate}
Through this example, agent $A_2$ discovered the satcuts $[3.5 - \cskew, 3.5]$ and $[5.8 - \cskew, 5.8]$.
The first is the leftmost satcut of the interval of satcuts.
$A_1$ discovered an additional satcut $[6, 6 - \cskew]$.
Joining this satcut with $A_2$'s second satcut returns a result of $[6, 5.8]$, which is the rightmost satcut of the interval of satcuts.

\section{Case Studies and Evaluation}
\label{sec:experiments}

We implemented our detection algorithm and ran experiments to
1) illustrate its operation, and
2) observe runtime scaling with number of agents and with average rate of events.
The detector was implemented in Julia for ease of prototyping, but future versions will be in C for speed.
All experiments are replicated to exhibit $95\%$ confidence interval.
Experiments were run on a single thread of an Ubuntu machine powered by an AMD Ryzen 7 5800X CPU @ 3.80GHz.
Code can be found at \url{https://github.com/sabotagelab/phryctoria}.

We consider two sources of data:
the first is a set of $N$ synthetically generated signals, $N=1...6$.
Each signal has a $5s$ duration, and is generated randomly while ensuring an average root rate of $\mu_n$.
That is, on average, $\mu_n$ roots exist in every second of signal $x_n$.
For the second source of data, we use the Fly-by-Logic toolbox~\cite{PantAM17CCTA} to control up to 6 simulated UAVs (i.e., drones) performing various reach-avoid missions.
Their 3-dimensional trajectories are recorded over 6 seconds.
We monitor the predicate ``All UAVs are at a height of at least $10m$ simultaneously''.
Maximum clock skew $\cskew$ is set to 0.05s.

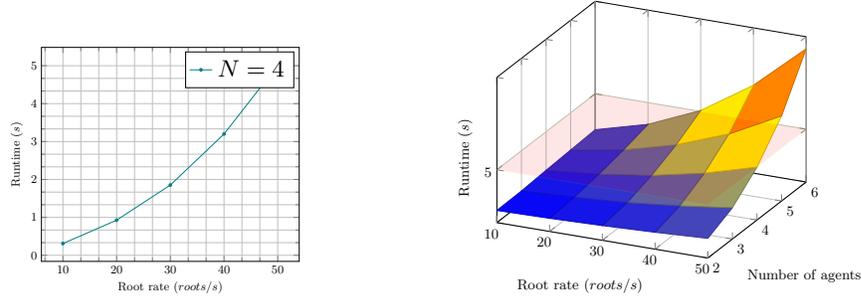
\begin{figure}[!ht]
	\begin{subfigure}[t]{0.45\textwidth}
		\centering
		\scalebox{0.5}{\begin{tikzpicture}
\begin{axis}[
    grid=both,
    minor tick num=2,
    xlabel={Root rate ($roots/s$)},
    ylabel={Runtime ($s$)},
    legend style={nodes={scale=2.0, transform shape}}
    ]

    \addplot[color=teal, mark=*, ultra thin, mark options={scale=0.5}] coordinates {
        (10, 0.310)
        (20, 0.927)
        (30, 1.853)
        (40, 3.199)
        (50, 5.013)
    };
    \addlegendentry{$N = 4$}
\end{axis}
\end{tikzpicture}}
		\caption{Runtime vs root rate on 4 synthetic signals.}
		\label{fig:rr_vs_rt}
	\end{subfigure}
	\hfill
	\begin{subfigure}[t]{0.45\textwidth}
		\centering
		\scalebox{0.6}{\begin{tikzpicture}
\begin{axis}[
    xlabel={\small Root rate $(roots/s)$},
    ylabel={\small Number of agents},
    zlabel={\small Runtime $(s)$},
    xtick = {10,20,30,40,50},
    ytick = {2,3,4,5,6},
    ztick = {5},
    grid,
    ]
    \addplot3[
        surf,
    ]
    coordinates {
        (10,2,0.06) (10,3,0.117) (10,4,0.31) (10,5,0.475) (10,6,0.634)

        (20,2,0.18) (20,3,0.462) (20,4,0.927) (20,5,1.56) (20,6,2.365)

        (30,2,0.373) (30,3,0.858) (30,4,1.853) (30,5,3.262) (30,6,5.119)

        (40,2,0.636) (40,3,1.542) (40,4,3.199) (40,5,5.65) (40,6,9.251)

        (50,2,0.974) (50,3,2.389) (50,4,5.013) (50,5,8.92) (50,6,14.751)
    };
    \addplot3[surf, red, opacity=0.1] coordinates {
        (10, 2, 5) (10, 6, 5)

        (50, 2, 5) (50, 6, 5)
    };
\end{axis}
\end{tikzpicture}}
		\smallcaption{Online monitoring. The red horizontal plane indicates the runtime threshold (namely, 5s) below which it is possible to do online detection.}
		\label{fig:online_monitoring}
	\end{subfigure}
	\caption{Runtime vs root rate and $N$ on synthetic data.}
	\label{fig:runtime vs root rate}
\end{figure}

\textbf{\bf Effect of root rate ($\mu_n$) on runtime.}
We use $4$ synthetic signals of $5s$ duration, and measure the detection runtime as the root rate for all signals is varied between $10 roots/s$ and $50 roots/s$.
\autoref{fig:rr_vs_rt} shows the results.
Naturally, as $\mu_n$ increases, so does the runtime due to having to process more tokens.

\paragraph*{\bf Online detection.} We want to identify when it is possible for us to perform online detection with the Julia implementation, i.e. such that the detector finishes before the end of the signal being processed.
To this end, we use the synthetic signals of duration 5s and vary both root rate and number of agents.
\autoref{fig:online_monitoring} shows the results:
all combinations of root rates and number of agents with runtimes under the threshold of 5s can be performed online with the hardware setup used for these experiments.

\begin{figure}[t]
	\begin{subfigure}[t]{0.45\textwidth}
		\centering
		\begin{tikzpicture}
\begin{axis}[
    width=\columnwidth,
    grid=both,
    minor tick num=2,
    xlabel={\small Number of agents},
    ylabel={\small Runtime ($s$)},
    legend style={nodes={scale=0.8, transform shape}}
    ]

    \addplot[color=teal, mark=*, ultra thin, mark options={scale=0.5}] coordinates {
        (2, 0.974)
        (3, 2.389)
        (4, 5.013)
        (5, 8.920)
        (6, 14.751)
    };
    \addlegendentry{$\mu_n = 50 roots/s$}
\end{axis}
\end{tikzpicture}
		\caption{Detection of synthetic signals at 50 roots/s}
		\label{fig:sig_vs_rt}
	\end{subfigure}
	\hfill
	\begin{subfigure}[t]{0.45\textwidth}
		\centering
		\begin{tikzpicture}
\begin{axis}[
    width=\columnwidth,
    grid=both,
    minor tick num=2,
    xlabel={\small Number of agents},
    ylabel={\small Runtime ($s$)},
    legend style={nodes={scale=0.8, transform shape}}
    ]

    \addplot[color=teal, mark=*, ultra thin, mark options={scale=0.5}] coordinates {
        (2, 0.031)
        (3, 0.067)
        (4, 0.193)
        (5, 0.262)
        (6, 0.318)
    };
\end{axis}
\end{tikzpicture}
		\smallcaption{Detection of UAV signals}
		\label{fig:sig_vs_uav}
	\end{subfigure}
	\caption{Runtime vs number of agents.}
	\label{fig:runtime vs N}
\end{figure}

\textbf{Effect of number of agents on runtime.}
\autoref{fig:runtime vs N} shows the effect of number of agents $N$ on runtime.
As expected, the runtime increases with $N$.

\section{Conclusion}
\label{sec:conclusion}
We have defined the first decentralized algorithm for continuous-time
predicate detection in partially synchronous distributed CPS.
To do so we analyzed the structure
of satisfying consistent cuts for conjunctive predicates, introduced a new notion of clock, and modified a classical discrete-time predicate detector.

\bibliographystyle{splncs04}
\bibliography{references}

\begin{thebibliography}{10}
\providecommand{\url}[1]{\texttt{#1}}
\providecommand{\urlprefix}{URL }
\providecommand{\doi}[1]{https://doi.org/#1}

\bibitem{basin2015failure}
Basin, D., Klaedtke, F., Z{\u{a}}linescu, E.: Failure-aware runtime
  verification of distributed systems. In: 35th IARCS Annual Conference on
  Foundations of Software Technology and Theoretical Computer Science (FSTTCS
  2015). vol.~45, pp. 590--603. Schloss Dagstuhl-Leibniz-Zentrum f{\"u}r
  Informatik (2015)

\bibitem{bataineh2019efficient}
Bataineh, O., Rosenblum, D.S., Reynolds, M.: Efficient decentralized ltl
  monitoring framework using tableau technique. ACM Transactions on Embedded
  Computing Systems (TECS)  \textbf{18}(5s),  1--21 (2019)

\bibitem{bauer2016decentralised}
Bauer, A., Falcone, Y.: Decentralised ltl monitoring. Formal Methods in System
  Design  \textbf{48},  46--93 (2016)

\bibitem{charron1991concerning}
Charron-Bost, B.: Concerning the size of logical clocks in distributed systems.
  Information Processing Letters  \textbf{39}(1),  11--16 (1991)

\bibitem{chauhan2013distributed}
Chauhan, H., Garg, V.K., Natarajan, A., Mittal, N.: A distributed abstraction
  algorithm for online predicate detection. In: 2013 IEEE 32nd International
  Symposium on Reliable Distributed Systems. pp. 101--110. IEEE, Braga,
  Portugal (2013)

\bibitem{cooper1991consistent}
Cooper, R., Marzullo, K.: Consistent detection of global predicates. ACM
  SIGPLAN Notices  \textbf{26}(12),  167--174 (1991)

\bibitem{gargbook}
Garg, V.: Elements of Distributed Computing. John Wiley \& Sons (2002)

\bibitem{garg2020predicate}
Garg, V.K.: Predicate detection to solve combinatorial optimization problems.
  In: Proceedings of the 32nd ACM Symposium on Parallelism in Algorithms and
  Architectures. pp. 235--245 (2020)

\bibitem{garg01slicing}
Garg, V.K., Mittal, N.: On slicing a distributed computation. In: Proceedings
  of the 21st International Conference on Distributed Computing Systems
  {(ICDCS} 2001), Phoenix, Arizona, USA, April 16-19, 2001. pp. 322--329.
  {IEEE} Computer Society (2001). \doi{10.1109/ICDSC.2001.918962},
  \url{https://doi.org/10.1109/ICDSC.2001.918962}

\bibitem{singhalbook}
Kshemkalyani, A., Singhal, M.: Distributed Computing: Principles, Algorithms,
  and Systems. Cambridge University Press (2011)

\bibitem{lamport1978time}
Lamport, L.: Time, clocks, and the ordering of events in a distributed system.
  Commun. ACM  \textbf{21}(7),  558–565 (7 1978).
  \doi{10.1145/359545.359563}, \url{https://doi.org/10.1145/359545.359563}

\bibitem{mattern1989virtual}
Mattern, F., et~al.: Virtual time and global states of distributed systems.
  Univ., Department of Computer Science, D 6750 Kaiserslautern, Germany (1989)

\bibitem{ntp}
Mills, D., Martin, J., Burbank, J., Kasch, W.: Network time protocol version 4:
  Protocol and algorithms specification. Tech. rep., Internet Engineering Task
  Force (2010)

\bibitem{momtaz23iccps}
Momtaz, A., Abbas, H., Bonakdarpour, B.: Monitoring signal temporal logic in
  distributed cyber-physical systems. In: Proceedings of the ACM/IEEE 14th
  International Conference on Cyber-Physical Systems (with CPS-IoT Week 2023).
  p. 154–165. ICCPS '23, Association for Computing Machinery, New York, NY,
  USA (2023). \doi{10.1145/3576841.3585937},
  \url{https://doi.org/10.1145/3576841.3585937}

\bibitem{momtaz2021predicate}
Momtaz, A., Basnet, N., Abbas, H., Bonakdarpour, B.: Predicate monitoring in
  distributed cyber-physical systems. In: International Conference on Runtime
  Verification. pp. 3--22. Springer, Online (2021)

\bibitem{mostafa2015decentralized}
Mostafa, M., Bonakdarpour, B.: Decentralized runtime verification of ltl
  specifications in distributed systems. In: 2015 IEEE International Parallel
  and Distributed Processing Symposium. pp. 494--503. IEEE (2015)

\bibitem{PantAM17CCTA}
Pant, Y.V., Abbas, H., Mangharam, R.: Smooth operator: Control using the smooth
  robustness of temporal logic. In: 2017 IEEE Conference on Control Technology
  and Applications (CCTA). pp. 1235--1240. IEEE (2017)

\bibitem{sen2004detecting}
Sen, A., Garg, V.K.: Detecting temporal logic predicates in distributed
  programs using computation slicing. In: Principles of Distributed Systems:
  7th International Conference, OPODIS 2003, La Martinique, French West Indies,
  December 10-13, 2003, Revised Selected Papers 7. pp. 171--183. Springer
  (2004)

\bibitem{sen2004efficient}
Sen, K., Vardhan, A., Agha, G., Rosu, G.: Efficient decentralized monitoring of
  safety in distributed systems. In: Proceedings. 26th International Conference
  on Software Engineering. pp. 418--427. IEEE (2004)

\end{thebibliography}
\appendix
\section{Proofs}

\subsection{Proof of Proposition~\ref{thm:cont lattice}}
Define the intersection $I = C\cap C'$ and let $e$ be an element of $I$.
Then by definition of a cut, every event that happened-before $e$ is in $C$ and in $C'$, and therefore is in their intersection, so $I$ is a cut.
To show that it satisfies the predicate, note that the frontier of $I$ is made of events $(t_n,x_n(t_n))$ such that $t_n = \max \{t~|~ (t,x_n(t)) \in C \cap C'\}$.
In words, $(t_n,x_n(t_n))$ is the last event on signal $x_n$ belonging to both satcuts, which implies it is the last event on at least one of the cuts, say $C$.
Therefore $(t_n,x_n(t_n))$ is on the frontier of $C$, and so  $x_n(t_n) \geq 0$ by definition of a conjunctive predicate.
Since this is true for every $n$ in $[N]$, we have that the frontier of $I$ is a consistent state that satisfies the predicate, and so $I \models \formula$.

The union $C\cup C'$ is also a satcut by similar arguments, so the set of satcuts is a lattice.

\subsection{Proof of Lemma \ref{thm:satcut}}
Let $C$ be a satcut, so that $x_n(t_n) \geq 0$ for every $(t_n,x_n(t_n))$ in its frontier.
Let $s_n$ be the biggest shift backwards in time preserving positivity:
\begin{equation}\label{eq:sn}
	s_n \defeq \sup\{s \such s \geq 0 \text{ and } \forall~0\leq \sigma \leq s.~ x_n(t_n-\sigma)\geq 0\}.
\end{equation}
By the starvation-freedom assumption \autoref{ass:basic}.\ref{ass:starvation freedom}, $s_n$ is finite and by the right-continuity of $x_n$, $x_n(t_n-s_n)\geq 0$.
Now the cut with frontier $(\evnt{n}{t_n - s_n})_n$ satisfies the predicate, but might not be consistent because it could be that $|t_n-s_n - (t_m-s_m)|>\cskew$ for some $n,m$.
Suppose without loss of generality that $t_1 - s_1$ is the largest of all the $t_n-s_n$'s.
Define $b_n = \max(t_n-s_n,t_1-s_1-\cskew)$ for all $n>1$.
Note that $b_n\leq t_n$ because $C$ is consistent\footnote{Indeed, $t_1-t_n\leq \cskew$ so a fortiori $t_1-t_n \leq \cskew+s_1$ and so $b_n = t_1-s_1-\cskew \leq t_n$. The other case, $b_n=t_n-s_n$, is immediate.}, and $x_n(b_n)\geq 0$.
Then the cut $L$ with frontier $(\evnt{1}{t_1-s_1}, \evnt{2}{b_2},\ldots,\evnt{N}{b_N})$ is consistent and satisfies the predicate.
It is also leftmost by construction of $s_1$.
Therefore $L$ is a leftmost satcut.

The reasoning for rightmost cuts follows the above lines, except for predicate satisfaction.
Namely: let $s_n$ now be the biggest shift forwards in time preserving positivity:
\begin{equation}\label{eq:sn rightmost}
	s_n \defeq \sup\{s~|~s\geq0 \text{ and } \forall~0\leq \sigma \leq s.~x_n(t_n+\sigma)\geq 0\}.
\end{equation}
By the starvation-freedom assumption \autoref{ass:basic}.\ref{ass:starvation freedom}, $s_n$ is finite.
Now the cut with frontier $(\evnt{n}{t_n + s_n})_n$ might not be consistent because it could be that $|t_n+s_n - (t_m+s_m)|>\cskew$ for some $n,m$.
Suppose without loss of generality that $t_1 + s_1$ is the smallest of all the $t_n+s_n$'s.
Define $b_n = \min(t_n+s_n,t_1+s_1+\cskew)$ for all $n>1$.
Note that $b_n\geq t_n$ because $C$ is consistent.
Then the cut $R$ with frontier $(\evnt{1}{t_1+s_1}, \evnt{2}{b_2},\ldots,\evnt{N}{b_N})$ is consistent, but does not necessarily satisfy the predicate because of possible discontinuities.
(Namely, if $t_n+s_n$ is a point of discontinuity for $x_n$ then possibly $x_n(t_n+s_n)<0$.)
$R$ is also rightmost by construction of $s_1$ and the $b_n$.
Therefore $R$ is a rightmost cut.

Thus every satcut is between a leftmost satcut and rightmost cut.
Also by construction of $L$ and $R$ (specifically, \autoref{eq:sn} and \ref{eq:sn rightmost}), there is no cut in-between that does not satisfy the predicate.
That is, there is no $C$ s.t. $L\sqsubseteq C \sqsubseteq R$ and $C\not\models \formula$. (Here $\sqsubseteq$ is the ordering relation on the lattice of cuts).

\subsection{Proof of \autoref{thm:finitely many satcuts}}
We will need the following definitions:
the \textit{leftmost event of a cut} $C$ is an event $\evnt{n}{t} \in \front(C)$ where $t \leq t'$ for all other events $\evnt{m}{t'} \in \front(C)$.
With $\beta$ a real number, an event $\evnt{m}{t'}$ is said to be $\beta$-offset from $\evnt{n}{t}$ if and only if $t' = t + \beta$.

We will need the following three lemmata.
\begin{lemma}
	\label{thm:rr of rightmost satcut}
	The leftmost event of a rightmost cut is a right root.
\end{lemma}
\begin{proof}
	Consider the leftmost event $\evnt{n}{t}$ of a rightmost cut $C$.
	Because $C$ is rightmost, then $x_n(t-\delta)\geq 0$ for all sufficiently small positive $\delta$.
	Assume for a contradiction that $\evnt{n}{t}$ is not a right root, so $x_n(t+\alpha)\geq 0$ for all sufficiently small $\alpha \geq 0$, say all $\alpha$ strictly less than some $\overline{\alpha}$.
	Since $\evnt{n}{t}$ is leftmost, we can add the events $\evnt{n}{t+\alpha}$, $0\leq \alpha \leq \min\{\frac{\cskew}{2},\frac{\overline{\alpha}}{2}, \gamma\}$, to $C$ to form a new  cut $C'$.
	Choosing $\gamma$ small enough guarantees that $C'$ is consistent.
	$C'$ is also satisfying because we only added events such that $x_n(t+\alpha)\geq 0$.
	This shows $C$ is not a rightmost cut, which contradicts our choice of $C$.
\end{proof}

\begin{lemma}
	\label{thm:satcut and roots}
	All events of the frontier of a rightmost cut are either right roots or $\cskew$-offset from a right root.
\end{lemma}
\begin{proof}
	Let $\evnt{n}{t}$ be the leftmost event in the frontier of a rightmost cut $C$.
	By \autoref{thm:rr of rightmost satcut} this event is a right root.
	Now consider any other event $\evnt{m}{t'} \in \front(C)$ which is \textit{not} a right root, and assume for contradiction that $t'\neq t+\cskew$.
	Then $x_m(t')\geq 0$ and (as in the proof of \autoref{thm:rr of rightmost satcut}) $x_m(t'+\alpha)\geq 0$ for all sufficiently small $\alpha$.
	If $t'<t+\cskew$, then it is possible to add the events $\{\evnt{m}{t'+\alpha} ~|~\alpha \in [0,\gamma)\}$ to $C$, with $\gamma$ small enough, to obtain a satcut to its immediate right, which contradicts $C$ being rightmost.
	On the other hand if $t'>t+\cskew$ this contradicts that $\evnt{n}{t}$ and $\evnt{m}{t'}$ are part of the same frontier.
	Thus $t'=t+\cskew$.
\end{proof}

The next lemma (and its proof) parallels \autoref{thm:rr of rightmost satcut} and \autoref{thm:satcut and roots}, but for leftmost satcuts.
\begin{lemma}
	\label{thm:lr of leftmost satcut}
	The rightmost event of a leftmost satcut is a left root.
	Moreover, every event of the frontier of a leftmost satcut is either a left root or is $(-\cskew$)---offset from a left root.
\end{lemma}

Thus every extremal satcut has a left root or a right root as one its constituent events.
Since there are only finitely many roots in any bounded interval, this gives us the desired conclusion.

\subsection{Proof of \autoref{thm:pvc isomorphic}}
The theorem's proof needs the following lemma, which is of independent interest, and which we prove first.
\begin{lemma}
	\label{thm:pvc hb geq}
	Let $n \neq m$ and $t, t' \neq 0$.
	Then $(\evnt{n}{t} \hb{} \evnt{m}{t'})$ iff $(\pvc{m}{t'}[n] \geq t)$.
\end{lemma}
\begin{proof}
	We split the bidirectional implication into its two directions:
	\begin{enumerate}
		\item $(\evnt{n}{t} \hb{} \evnt{m}{t'}) \implies (\pvc{m}{t'}[n] \geq t)$\\
		Since $\pvc{n}{t}[n] = t$ by \autoref{def:pvc}~\ref{def:pvc local} and $\evnt{n}{t} \hb{} \evnt{m}{t'}$, then by \autoref{def:pvc}~\ref{def:pvc hb}, $\pvc{m}{t'}[n] \geq t$.
		\item $(\evnt{n}{t} \hb{} \evnt{m}{t'}) \impliedby (\pvc{m}{t'}[n] \geq t)$
		\begin{enumerate}
			\item Case $(\pvc{m}{t'}[n] = t) \implies (\evnt{n}{t} \hb{} \evnt{m}{t'})$:\\
			Besides initialization, the only case in \autoref{def:pvc} where a value is assigned which did not come from another timestamp is \autoref{def:pvc}~\ref{def:pvc local}.
			Consider an event $\evnt{n}{t}$.
			The timestamp of this event at index $n$ is $t$, by \autoref{def:pvc}~\ref{def:pvc local}.
			At the point in time when this event is created (local time $t$ on agent $A_n$), no other timestamp has the value $t$ at index $n$.
			All other $\pvc{m}{t'}$ which have the value $t$ at index $n$ must be assigned by \autoref{def:pvc}~\ref{def:pvc hb}.
			This means that they have the relation $\evnt{n}{t} \hb{} \evnt{m}{t'}$, due to the transitive property of the happened-before relation.
			\item Case $(\pvc{m}{t'}[n] > t) \implies (\evnt{n}{t} \hb{} \evnt{m}{t'})$:\\
			Consider a $t''$ where $\pvc{m}{t'}[n] = t''$ and $t'' > t$.
			Then by the previous case, $\evnt{n}{t''} \hb{} \evnt{m}{t'}$.
			Since by the happened-before relation all events on an agent are totally ordered (\autoref{def:dstr signal}~\ref{def:hb local}), $\evnt{n}{t} \hb{} \evnt{n}{t''}$.
			By the transitive property of the happened-before relation (\autoref{def:dstr signal}~\ref{def:hb transitive}), $\evnt{n}{t} \hb{} \evnt{m}{t'}$.
		\end{enumerate}
	\end{enumerate}
\end{proof}

We now proceed with the proof of the theorem.
\begin{proof}
	Since each PVC timestamp corresponds to exactly one event and all events have a timestamp, there is clearly a bijective mapping.
	To show it preserves order, we need to confirm that $(\evnt{n}{t} \hb{} \evnt{m}{t'}) \iff (\pvc{n}{t} < \pvc{m}{t'})$.
	\begin{enumerate}
		\item $\evnt{n}{t} \hb{} \evnt{m}{t'} \implies \pvc{n}{t} < \pvc{m}{t'}$\\
		By \autoref{def:pvc}~\ref{def:pvc hb}, each element of $\pvc{n}{t}$ must be less than or equal to the corresponding element of $\pvc{m}{t'}$.
		So then we need to show that $\pvc{n}{t} \neq \pvc{m}{t'}$.
		\autoref{def:pvc}~\ref{def:pvc local} indicates that $\pvc{m}{t'}[m] = t'$.
		By \autoref{thm:pvc hb geq} if $\pvc{n}{t}[m] = t'$ then $\evnt{m}{t'} \hb{} \evnt{n}{t}$; but there cannot be cycles in the happened-before order relation, then $\pvc{n}{t}[m] < t'$.
		This implies that $\pvc{n}{t} < \pvc{m}{t'}$.
		\item $(\evnt{n}{t} \hb{} \evnt{m}{t'}) \impliedby (\pvc{n}{t} < \pvc{m}{t'})$\\
		$\pvc{n}{t} < \pvc{m}{t'}$ means that $\pvc{n}{t}[i] \leq \pvc{m}{t'}[i],\: \forall i \in [N]$.
		Consider index $n$, where $\pvc{n}{t}[n] \leq \pvc{m}{t'}[n]$.
		By \autoref{def:pvc}~\ref{def:pvc local}, $\pvc{n}{t}[n] = t$, so $\pvc{m}{t'}[n] \geq t$.
		Then \autoref{thm:pvc hb geq} states that this implies $\evnt{n}{t} \hb{} \evnt{m}{t'}$.
	\end{enumerate}
\end{proof}

\subsection{Proof of \autoref{thm:pvc assign}}
\begin{proof}
	Consider \autoref{def:dstr signal}~\ref{def:hb eps}.
	This indicates that all events $\evnt{i}{t - \cskew}$ happened-before $\evnt{n}{t}$, $\forall i \in [N] \setminus \{n\}$.
	Therefore, if these events directly happened-before $\evnt{n}{t}$ (there is no $\evnt{m}{t'}$ where $\evnt{i}{t - \cskew} \hb{} \evnt{m}{t'}$ and $\evnt{m}{t'} \hb{} \evnt{n}{t}$), then this vector is a correct assignment.

	By looking at each point in \autoref{def:dstr signal}, we can see that the only case where one event happened-before another on a different process is when there is at least $\cskew$ difference, \autoref{def:hb eps}.
	While an event may have happened-before $\evnt{n}{t}$ by indirectly following \autoref{def:hb eps} by way of \ref{def:hb local} and \ref{def:hb transitive}, we do not need to consider this event because there is not a direct happened-before relation with $\evnt{n}{t}$ (no event in between).
	Therefore, the assignment $[t - \cskew, \dots, t - \cskew, t, t - \cskew, \dots, t - \cskew]$ is suitable for timestamp $\pvc{n}{t}$.
\end{proof}

\subsection{Proof of \autoref{thm: mod cgnm correctness}}
We start with two lemmas for the leftmost satcut events.
\begin{lemma}
	\label{thm:ls token events}
	For all events $\evnt{n}{t}$ that are left roots, the token $T_n$ incorporates all $\evnt{m}{t - \cskew}$ for all $m \neq n$.
\end{lemma}
\begin{proof}
	For a left root $\evnt{n}{t}$, by \autoref{thm:pvc assign} its PVC is $\pvc{n}{t} = [t - \cskew, \dots, t - \cskew, t, t - \cskew, \dots, t - \cskew]$.
	Since a token $T_n$ is tasked with identifying consistent cuts, for each $m \neq n$ it must incorporate the leastmost event on $A_m$ which can form a consistent cut with $\evnt{n}{t}$.
	The PVC identifies this event as $\evnt{m}{t - \cskew}$.
	Therefore, $T_n$ incorporates all $\evnt{m}{t - \cskew}$ events where $\evnt{n}{t}$ is a left root on $A_n$.
\end{proof}

\begin{lemma}
	\label{thm:ls events}
	The modified slicer processes all events of a leftmost satcut.
\end{lemma}
\begin{proof}
	By \autoref{thm:lr of leftmost satcut}, all events of a leftmost satcut are either at time $t$ or $t - \cskew$, where $t$ is the time of a left root.
	Since by \autoref{thm:ls token events} every token $T_n$ will visit the $t - \cskew$ event for any left root at $t$ on $A_n$, every $t - \cskew$ will be processed for any left root.
	Thus, all events of a leftmost satcut will be processed.
\end{proof}

We can now proceed with the proof of the theorem.
\begin{proof}
	The abstractor creates discrete events for all roots (Abstractor Line~\ref{line:abs:add root}), as well as $\cskew$-offsets from right roots.
	By \autoref{thm:ls events}, the slicer creates all events of a leftmost satcut.
	This means that all events of leftmost and rightmost satcuts are processed by the slicer.
	Therefore, since the modified slicer returns a lattice of satcuts, the extremal satcuts are included.
\end{proof}

\subsection{Proof of \autoref{thm:time complexity}}
\underline{Time complexity}.
The calculations in our algorithm come from the abstractor, and the modification to the \slicer{}.
Finding a root of a signal $x_n$ takes constant time in the system parameters.
The abstractor has every process send right root info to every other process, for a complexity of $N-1$ per right root, and total complexity of $(N-1)R$ where $R$ is the number of right roots in the system in a given bounded window of time.

Consider slicer $\ourslicer_n$, which is hosting token $T_m$.
The slicer creates a new event, for every target event of $T_m$ that was not produced by the abstractor of $A_m$.
Event creation is $O(N)$ since it requires the creation of a size-$N$ PVC assigned to the event.
Event storage takes constant time if the new event is simply appended at the end of the local buffer, or $O(k)$ if the event is inserted in-order in the sorted local buffer of size $k$.
Either one works: the first one is cheaper, but an unsorted buffer costs more to find events in it.
The latter is more expensive up-front, but the sorted buffer can be searched faster.
Either way, the slicer modification costs a total of $O(N\cdot M)$ in a given bounded window of time with $M$ missed events in the system.

Now the number of target events requiring creation is on the order of the number of right roots since they result from left roots, and there are equal numbers of left and right roots. Thus $M = O(R)$.
Therefore, the total complexity for our algorithm in a given bounded window of time is $O(R(N-1+N))$.
Of course, this is then added to the complexity of running the modified slicer, which is $O(N^2D)$, where $D$ is the number of events in the discrete-time signal. At the most, there are $2R$ events.
So finally the total time complexity  is $O(R(N-1+N)+2N^2R)$, or $O(R(2N+2N^2)/N) = O(2RN)$ per agent.

\underline{Space complexity.} Indeed, a PVC timestamp has size $O(N)$ (since it's an $N$-dimensional vector).
This is in fact the optimal complexity for characterizing causality \cite{charron1991concerning}.
One token stores $N$ PVCs at all times and token updates replace old PVC values by new ones.
Therefore one token has size $O(N^2)$, and all $N$ tokens (one per agent) require $O(N^3)$ space.

How long events stay in the abstractor's local buffers depends on message transmission times, since events are removed from the buffers after the appropriate messages are received (see Algorithm~\ref{algo:abstractor} lines \ref{line:abs:add root},\ref{line:abs:add virtual},\ref{line:abs:remove event}).
It also depends on the distribution of events within the interval of analysis, not just their rate $1/R$.
E.g. if roots are uniformly distributed in the analysis interval, then the $n^{th}$ abstractor's local buffer grows at the most to size $O(N^2)$, as it receives roots from the other $N-1$ agents and stores the $O(N)$ PVC timestamp for each root.
Then event removal starts as $A_n$ receives target events.
Similar considerations apply to the slicer's local buffers.
In such a case the detector's total space complexity is $O(N^3+2N^2)$.
\end{document}